\documentclass[USenglish, cleveref, autoref, unicode, runningheads]{llncs}
\usepackage[T1]{fontenc}
\usepackage{graphicx}

\usepackage{latexsym}
\usepackage{amsmath,amssymb,mathtools}
\usepackage{amsfonts}
\usepackage{hyperref}
\usepackage{cite}
\usepackage[capitalize]{cleveref}
\usepackage{tabularx,environ}
\usepackage{comment}
\usepackage{url}
\usepackage[full]{complexity}
\usepackage{lscape}
\usepackage{fbox}

\usepackage[normalem]{ulem}
\usepackage{tikz}
\usepackage[mode=buildnew]{standalone}

\usepackage{diagbox}
\usepackage{mathdots}

\newenvironment{claimproof}{\begin{proof}}{\hfill $\blacksquare$\end{proof}}
\newcommand{\proofsubparagraph}[1]{\paragraph{#1}}
\crefname{algocf}{Algorithm}{Algorithms}
\crefname{observation}{Observation}{Observations}
\crefname{claim}{Claim}{Claims}

\usepackage[appendix=inline]{apxproof}

\newtheoremrep{theorem}{Theorem}
\newtheoremrep{lemma}[theorem]{Lemma}
\newtheoremrep{claim}[theorem]{Claim}
\newtheoremrep{observation}[theorem]{Observation}

\usepackage{algpseudocode}
\usepackage[linesnumbered,ruled,vlined]{algorithm2e}
\SetKwProg{Procedure}{Procedure}{}{}

\allowdisplaybreaks

\DeclarePairedDelimiter{\braces}{\{}{\}}            %
\DeclarePairedDelimiter{\parens}{(}{)}              %
\DeclarePairedDelimiter{\angles}{\langle}{\rangle}  
\DeclarePairedDelimiterX{\setdef}[2]{\{}{\}}{#1 : #2}        %
\DeclarePairedDelimiterXPP{\exclude}[1]{\mathopen{}\setminus}{\{}{\}}{}{#1}

\newcommand{\TMC}{\textsc{Tree Minor Containment}}
\newcommand{\PPM}{\textsc{Inclusive Poset Pair Cover}}
\newcommand{\InclusiveSetCover}{\textsc{Inclusive Set Cover}}

\newcommand{\dun}{\mathbin{\sqcup}}
\newcommand{\bigdun}{\operatorname{\bigsqcup}}

\newcommand{\Nset}{\mathbb N}

\newcommand{\diam}{\mathrm{diam}}

\newcommand{\defproblem}[3]{
    \vspace{1mm}
    \noindent\fbox[bt]{
        \begin{minipage}{0.97\textwidth}
            #1\newline
            {\bf{Input:}} #2\newline
            {\bf{Question:}} #3
        \end{minipage}
    }
}

\newcommand{\posetcat}{\operatorname{OCat}}
\newcommand{\ordercat}{order caterpillar}
\usepackage{color}

\begin{document}
\title{Dichotomies for Tree Minor Containment with Structural Parameters%
\thanks{%
    Partially supported by JSPS KAKENHI Grant Numbers %
    JP23KJ1066, %
    JP21J20547, %
    JP21K17812, %
    JP22H03549, %
    JP21K11752, and %
    JP22H00513, %
    and JST ACT-X Grant Number JPMJAX2105. %
}} %
\author{Tatsuya Gima\inst{1,3}\orcidID{0000-0003-2815-5699} \and
Soh Kumabe\inst{2,3}\orcidID{0000-0002-1021-8922} \and
Kazuhiro Kurita\inst{1}\orcidID{0000-0002-7638-3322} \and
Yuto Okada\inst{1}\orcidID{0000-0002-1156-0383} \and
Yota Otachi\inst{1}\orcidID{0000-0002-0087-853X}}
\institute{Nagoya University, Nagoya, Japan \\
\email{gima@nagoya-u.jp, kurita@i.nagoya-u.ac.jp, okada.yuto.b3@s.mail.nagoya-u.ac.jp, otachi@nagoya-u.jp} \and 
The University of Tokyo, Tokyo, Japan \\
\email{soh\_kumabe@mist.i.u-tokyo.ac.jp}
\and 
JSPS Research Fellow
}

\authorrunning{T. Gima, S. Kumabe, K. Kurita, Y. Okada, and Y. Otachi} %

\maketitle              %
\begin{abstract}
The problem of determining whether a graph $G$ contains another graph $H$ as a minor, referred to as the \emph{minor containment problem}, is a fundamental problem in the field of graph algorithms. While it is \NP-complete when $G$ and $H$ are general graphs, it is sometimes tractable on more restricted graph classes. 
This study focuses on the case where both $G$ and $H$ are trees, known as the \emph{tree minor containment problem}. Even in this case, the problem is known to be \NP-complete. In contrast, polynomial-time algorithms are known for the case when both trees are caterpillars or when the maximum degree of $H$ is a constant.
Our research aims to clarify the boundary of tractability and intractability for the tree minor containment problem. Specifically, we provide dichotomies for the computational complexities of the problem based on three structural parameters: the diameter, pathwidth, and path eccentricity.
\end{abstract}

\section{Introduction}
In the field of graph algorithms, given two graphs $G$ and $H$, the problem of determining whether $G$ contains $H$ is a fundamental problem.
This type of problem, such as (induced) subgraph isomorphism~\cite{Garey:book:1990}, minor containment~\cite{MATOUSEK1992343}, and topological embedding~\cite{LaPaugh:ATC:78}, is often \NP-complete when $G$ and $H$ are general graphs.
Therefore, extensive research has been conducted on whether these problems can be efficiently solved on more restricted classes of graphs~\cite{MATOUSEK1992343,Gupta:DAM:2005,HajiaghayiN:JCSS:07,BodlaenderHKKOO:Algorithmica:20}.
The class of trees is the most fundamental one among such graph classes.
For all the problems listed above, except the minor containment problem, there are polynomial time algorithms~\cite{Matula:AAC:1978,Shamir:JA:1999,Abboud:ACM:2018}, even for generalized versions~\cite{Gupta:Alg:1998}.

We focus on the \emph{minor containment problem}, which is the problem of determining whether graph $G$ contains graph $H$ as a minor. Even when both $G$ and $H$ are trees, in which case we call the problem \TMC{}, it remains \NP-complete~\cite{MATOUSEK1992343}. Furthermore, it remains \NP-complete even if the diameters of both trees are constant~\cite{MATOUSEK1992343}. However, polynomial-time algorithms are known for cases where the maximum degrees of $H$ is constant~\cite{Pekka:SIAM:1995,Akutsu:TCS:2021,Naomi:IJFCS:2000} or when both trees are caterpillars~\cite{Gupta:DAM:2005,Miyazaki:ICPRAM:2022}. 
Therefore, what condition makes \TMC{} tractable is a natural question.
In what follows, we denote $G$ and $H$ as $T$ and $P$, respectively, since both graphs are trees.

\subsection{Our Contributions}
\begin{table}[t]
    \centering
    \begin{minipage}[c]{0.35\textwidth}
    \centering
    \begin{tabular}{c|c|c|c|c}
    $\diam$ & $\le 3$ & \;$4$\; & $5$ & $\geq 6$ \\ \hline
    $\le 3$                           &  \multicolumn{3}{c}{\P}        \\ \cline{1-2}\cline{5-5}
    $4$                               &         & \multicolumn{2}{c|}{} & \NPC \\ \cline{1-1}\cline{3-3}
    $5$                               &  \multicolumn{2}{c|}{} &   &   \\ \cline{1-1}\cline{4-4}
    $\geq 6$                          &  \multicolumn{3}{c|}{meaningless}         \\
    \end{tabular}
    \label{tab:diameter}
    \end{minipage}%
    \begin{minipage}[c]{0.32\textwidth}
    \centering
    \begin{tabular}{c|c|c|c}
    $\mathrm{pe}$ & $\le 1$ & $2$ & $\geq 3$ \\  \hline
    $\le 1$ & \multicolumn{3}{c}{\P} \\ \cline{1-2}\cline{4-4}
    $2$     & & & \NPC \\ \cline{1-1}\cline{3-3}
    $\geq 3$ & \multicolumn{2}{c|}{meaningless} & \\ 
    \end{tabular}
    \end{minipage}%
    \begin{minipage}[c]{0.32\textwidth}
    \centering
    \begin{tabular}{c|c|c|c}
    $\mathrm{pw}$ & \; $1$ \;& $2$ & $\geq 3$ \\  \hline
    $1$ & \multicolumn{3}{c}{\P} \\ \hline
    $2$ & &\multicolumn{2}{c}{\NPC}\\ \cline{1-1}\cline{3-3}
    $\geq 3$ & \multicolumn{2}{c|}{meaningless} & \\ 
    \end{tabular}
    \end{minipage}%
    \caption{In these tables, $\diam$, $\mathrm{pe}$, and $\mathrm{pw}$ denote the diameter, path eccentricity, and pathwidth, respectively.
    The first row represents the values that a tree $T$ has, and the first column represents the values that a tree $P$ has.
    There is no need to consider problems in these areas marked ``meaningless''.}
    \label{tab:summary}
\end{table}

In this paper, we show dichotomies for three different structural parameters, diameter, pathwidth, and path eccentricity. We summarize dichotomies with respect to each parameter in \Cref{tab:summary}.

Even when the diameters of $T$ and $P$ are constant, it is known that \TMC{} is \NP-complete~\cite{MATOUSEK1992343}.
Although they did not clarify the exact value of the constant, it can easily be observed that the constant is $8$, which is not tight.
Our first contribution is to provide the tight diameter requirement for \TMC{} to be \NP-complete.

\begin{theorem}
    \TMC{} is \NP-complete if the diameters of $T$ and $P$ are at least $6$ and $4$, respectively.
    Otherwise, \TMC{} can be solved in polynomial time.
\end{theorem}

When the pathwidths of both trees are $1$ (or equivalently, both trees are \emph{caterpillars}), \TMC{} can be solved in polynomial time~\cite{Gupta:DAM:2005,Miyazaki:ICPRAM:2022}.
Our second contribution is extending the positive result to the case where the pathwidth of $T$ is arbitrary, and proving tight \NP-completeness.
\begin{theorem}
    \TMC{} is \NP-complete if the pathwidths of both trees are at least $2$.
    Otherwise, \TMC{} can be solved in polynomial time.
\end{theorem}

As evident from the theorem above, a caterpillar is an important class to consider when studying the tractability of \TMC{}.
The \emph{path eccentricity} is known as a more direct parameter to express ``caterpillar-likeness,''~\cite{Gomez:DAM:2023, Mitchell:TS:1981} which is defined as the distance from a specific path to the farthest vertex. 
The path eccentricity of a caterpillar is $1$, and a tree of a path eccentricity $2$ is called a \emph{lobster}. 
Our third contribution is the following.
\begin{theorem}
    \TMC{} is \NP-complete if the path eccentricities of $T$ and $P$ are at least $3$ and $2$, respectively. 
    Otherwise, \TMC{} can be solved in polynomial time.
\end{theorem}
By definition, for a tree, the path eccentricity is at most the pathwidth. Therefore, the positive result for the case where both $T$ and $P$ have path eccentricity of $2$ can be seen as encompassing cases that were not covered by considering the dichotomy for pathwidth.

\subsection{Related Work}

The most significant result concerning the minor containment problem is probably the Graph Minor Theory developed by Robertson and Seymour~\cite{Robertson:CT:95}. They proved that the minor containment problem can be solved in $f(H)\cdot O(|V(G)|^3)$-time, where $f$ is some computable function. Using this algorithm, they proved the existence of an algorithm that determines whether a graph $G$ satisfies any minor-closed property in $O(|V(G)|^3)$ time. Kawarabayashi, Kobayashi, and Reed improved this time complexity to $O(|V(G)|^2)$~\cite{kawarabayashi2012disjoint}.

Matou\v{s}ek and Thomas proved that this problem remains \NP-complete even on trees with bounded diameters~\cite{MATOUSEK1992343}. Furthermore, they addressed the minor containment problem on graphs with treewidth $k$, and provided a polynomial-time algorithm for cases where $H$ is connected and its  degree is bounded and 
Gupta et al. provided a polynomial-time algorithm for the case where both $G$ and $H$ are $k$-connected and have pathwidth at most $k$~\cite{Gupta:DAM:2005}. Their results can also be applied to the subgraph isomorphism problem and the topological embedding problem.

A generalization of \TMC{}, called the \emph{tree inclusion problem}, has also been investigated. In this problem, we are given two rooted trees, $T$ and $P$, with labeled vertices, and the objective is to determine whether it is possible to repeatedly contract vertices of $T$ towards their parent until $T$ matches $P$, including the labels. The special case where all vertices have the same label corresponds to \TMC{}. Kilpel\"{a}inen and Mannila showed that there is an \FPT-time algorithm parameterized by the maximum degree of $P$.
It runs in $O(4^{\deg(P)}\cdot \mathrm{poly}(n))$ time~\cite{Pekka:SIAM:1995}, and
Akutsu et al.\ improved this result to $O(2^{\deg(P)}\cdot \mathrm{poly}(n))$ time, where $\deg(P)$ is the maximum degree of $P$~\cite{Akutsu:TCS:2021}.
Miyazaki, Hagihara, and Hirata have provided a polynomial-time algorithm for the case where both $T$ and $P$ are caterpillars \cite{Miyazaki:ICPRAM:2022}. Additionally, Kilpel\"{a}inen and Mannila have proved that the problem remains \NP-complete even when $T$ has depth $3$ \cite{Pekka:SIAM:1995}. However, it should be noted that their proof relies on the existence of labels, so it does not directly imply our \NP-completeness result for the \TMC{} for trees with bounded diameters.

As another generalization of \TMC{}, the problem of finding the smallest tree containing two trees as minors is also investigated. For this problem, Nishimura, Ragde, and Thilikos gave an \FPT-time algorithm parameterized by the maximum degree~\cite{Naomi:IJFCS:2000}.

\section{Preliminaries}
Let $T$ be a tree and $n$ be the number of vertices or nodes in $T$.
We denote the set of vertices and edges of $T$ as $V(T)$ and $E(T)$, respectively.
For a vertex $v$, the set of vertices adjacent to $v$ is the \emph{neighbors of $v$} and denoted by $N_T(v)$.
The cardinality of the neighbor of $v$ is the \emph{degree of $v$} and is denoted by $\deg_G(v)$.
Moreover, the \emph{degree of $G$} is defined by $\max_{v \in V} \deg_G(v)$ and denoted by $\deg(G)$.
For two vertices $u, v \in V$, the \emph{distance} between $u$ and $v$ is the length of a shortest $u$-$v$ path.
We denote the distance between $u$ and $v$ as $\mathrm{dist}(u, v)$.
The \emph{diameter} of a tree $T$, denoted by $\diam(T)$, is the maximum distance between two vertices in $T$.
For a set of edges $F$, we denote an edge-induced subgraph $T[E \setminus F]$ as $T - F$.
Similarly, we denote an induced subgraph $T[V \setminus U]$ as $T - U$.
For a tree $T$ and a set of vertices $U$, 
\emph{vertex contraction} $T \slash U$ is the graph obtained by considering all vertices in $U$ identical.
More precisely, $V(T \slash U) \coloneqq (V \setminus U) \cup \{w\}$ and $E(T\slash U) \coloneqq \{ \{u, v\} \mid \{u, v\} \in E(T)\ \land u, v \in V(T \slash U)\} \cup \{ \{w, v\} \mid v \in V(T \slash U) \land \exists u \in U, \{u, v\} \in E(T)\}$.
For two disjoint trees $T = (V, E)$ and $P = (U, F)$, we denote the forest $(V \cup U, E \cup F)$ as $T \cup P$.

A tree $T$ is \emph{caterpillar} if $T$ becomes a path by removing all leaves in $T$.
Moreover, $T$ is \emph{lobster} if $T$ becomes a caterpillar by removing all leaves in $T$.
As a generalization of lobsters, a tree $T$ is \emph{$k$-caterpillar} if $T$ becomes a path by removing all leaves $k$ times.
We call the minimum value of $k$ \emph{path eccentricity} of $T$.
Therefore, $T$ is a path if and only if $k = 0$, $T$ is a caterpillar if and only if $k \le 1$, and $T$ is a lobster if and only if $k \le 2$.
A path $P$ is a \emph{backbone} of a $k$-caterpillar $T$ if for any $v \in T$, $P$ has a vertex $u$ such that $\mathrm{dist}(u, v) \le k$.

We next define the \emph{pathwidth} of $T = (V, E)$.
The pathwidth of $T$ is defined by a \emph{path decomposition} of $T$.
A path decomposition of $T$ is a pair $(\mathcal X, P)$, where $P = (V_P, E_P)$ is a path and $\mathcal X = \{\mathcal X_i \mid i \in V_P \}$ is a family of subsets of $V$, called \emph{bags} that satisfies the following conditions.
(I)   $\bigcup_{i \in V_P} \mathcal X_i = V$,
(II)  for each edge $e \in E$, there is a bag $\mathcal X_i$ such that $T[\mathcal X_i]$ contains $e$, and
(III) for all $v \in V$, we define the set of vertices $U \coloneqq \{ i \in V_P \mid v \in \mathcal X_i \}$ and $P[U]$ is connected.
For a path decomposition $(\mathcal X, P)$, the width of this decomposition is defined by $\max_{\mathcal X_i \in \mathcal X} |\mathcal X_i| - 1$.
Moreover, the \emph{pathwidth} of $T$ is the minimum width of any path decomposition.
We denote it as $\mathrm{pw}(T)$.

A tree $P$ is a \emph{minor} of a tree $T$ if there exists a surjective map called \emph{minor embedding} (or simply embedding) $f\colon T\to P$ such that
\begin{itemize}
    \item for all $v\in V(P)$, the subgraph of $T$ induced by $f^{-1}(v)$ is connected, and
    \item for all $e=(u,v)\in E(P)$, there exists an edge $e'=(u',v')$ of $T$ such that $f(u')=u$ and $f(v')=v$.
\end{itemize}
If $P$ is a minor of $T$, we say that $T$ \emph{contains} $P$ as a minor.

Finally, we give the definition of the problem addressed in this paper.

\defproblem{\TMC}{Two trees $T$ and $P$.}{Is $P$ a minor of $T$?}

Theorems and lemmas marked with ($\ast$) are shown in the appendix due to space limitation.
\section{\NP-completeness of \TMC{}}
We show that \TMC{} is \NP-complete even if 
diameters of $T$ and $P$ are at least $6$ and $4$, respectively, or
pathwidths of $T$ and $P$ are at least $2$.
In \Cref{subsec:hardness:diameter}, we show that \TMC{} is \NP-complete 
if diameters of $T$ and $P$ are at least $6$ and $4$, respectively.
Moreover, in \Cref{subsec:hardness:pw}, we show that \TMC{} is \NP-complete if pathwidths of $T$ and $P$ are at least $2$.

\subsection{Bounded Diameter and Bounded Path Eccentricity}\label{subsec:hardness:diameter}

In this subsection, we improve the previous bound in \cite{MATOUSEK1992343}.
To this end, we show the \NP-completeness of \InclusiveSetCover{}, a variant of \textsc{Set Cover}.
To define \InclusiveSetCover, we introduce some notations.
The \emph{disjoint union} of two sets $A$ and $B$ is, denoted by $A \dun B$, $\setdef{(a, 0)}{a\in A}\cup \setdef{(b,1)}{b\in B}$.
The \emph{disjoint union} of a family of sets $\mathcal A = (A_i)_{i\in \lambda}$ is $\bigcup_{i\in \lambda}\setdef{(a, i)}{a\in A_i}$, and denoted by $\bigdun_{i \in \lambda} A_i$ or simply $\bigdun \mathcal A$.
We often consider an element $(x,i) \in A \dun B$ (or $(x,i) \in \bigdun_{i \in \lambda} A_i$) simply as an element $x \in A \cup B$ (or $x \in \bigcup_{i \in \lambda} A_i$ respectively).
We are ready to define \InclusiveSetCover{}.

\defproblem{\InclusiveSetCover}{A set $U = \{1, 2, \dots, n\}$, a collection of $m$ sets $\mathcal{S} \subseteq 2^{U}$, and an integer $k \in \Nset$.}{Does there exist $\mathcal{R} \subseteq \mathcal{S}$ such that $|\mathcal{R}| \leq k$ and there is a surjection $f \colon \bigdun \mathcal R \to U$ such that $v \geq f((v,i))$ for each $(v,i) \in \bigdun\mathcal R$?}

\begin{lemmarep}
    \InclusiveSetCover~is \NP-complete.
\end{lemmarep}
\begin{proof}
It is clear that this problem is in \NP.
To show the \NP-hardness, we give a reduction from \textsc{3-SAT}, which is known to be \NP-hard\cite{Garey:book:1990}.

Let us consider an instance $\angles{V, \mathcal{C}}$ of \textsc{3-SAT}, where $V$ is a set of variables $\{x_1, x_2, \dots, x_{|V|}\}$ and $\mathcal{C}$ is a set of clauses $\{C_1, C_2, \dots, C_{|\mathcal{C}|}\}$.
From this instance, we reduce to an instance $\angles{U, \mathcal{S}, k}$ of \InclusiveSetCover~in the following way.
Let $U$ be the set $\{1, 2, \dots, 2|V| + 3|\mathcal{C}|\}$.
For each integers $x_i$, we define the set of integers $T_i, F_i$ as follows, where $\alpha =|V| + 3|\mathcal{C}|$.
Let $T_i$ be $\{\alpha-i+1, \alpha+i\} \cup \{3j \mid x_i \in C_j\} \cup \{3j-1 \mid \bar{x}_i \in C_j\}$ and $F_i$ be $\{\alpha-i+1, \alpha+i\} \cup \{3j \mid \bar{x_i} \in C_j\} \cup \{3j-1 \mid x_i \in C_j\}$.
Let $\mathcal{S}$ be the collection of the above sets $\{T_1, T_2, \dots, T_{|V|}, F_1, F_2, \dots, F_{|V|}\}$.
Lastly, let $k$ be $|V|$.
The above instance can be constructed in polynomial time.

Note that we use three integers $\{3j - 2, 3j-1, 3j\} \subseteq U$ for each clause $C_j$ and $\{\alpha-i+1, \alpha+i\}$ for each variable $x_i$.
In total, we use $\{1, 2, \dots, 3|\mathcal{C}|\}$ for the clauses and $\{3|\mathcal{C}|+1, 3|\mathcal{C}|+2, \dots, 2|V| + 3|\mathcal{C}|\} = \{\alpha - |V| + 1, \alpha - |V| + 2, \dots, \alpha + |V|\}$ for the variables.

\proofsubparagraph{Completeness}

\newcommand{\AssignmentVariable}{\varphi}
\newcommand{\ValueTrue}{\mathrm{true}}
\newcommand{\ValueFalse}{\mathrm{false}}

We show that if $\angles{V, \mathcal{C}}$ is satisfiable, then $\angles{U, \mathcal{S}, k}$ is a yes-instance.
Let $\AssignmentVariable \colon V \to \{\ValueFalse, \ValueTrue\}$ be an assignment that satisfies $\angles{V, \mathcal{C}}$.

From $\AssignmentVariable$, we obtain the solution $\mathcal{R}$ for $\angles{U, \mathcal{S}, k}$ by selecting $T_i$ if $\AssignmentVariable(x_i)$ is $\ValueTrue$ and $F_i$ otherwise for each $x_i$.
Since $\AssignmentVariable$ satisfies each clause $C_j$, we have three integers from $\{3j-1, 3j\}$ and at least one of them must be $3j$.
Allocating them to $\{3j-2, 3j-1, 3j\}$, we cover $\{1, 2, \dots, 3|\mathcal{C}|\}$.
In addition, since we have selected $T_i$ or $F_i$ for all $i$, $|\mathcal{R}|$ is $k$ and we cover $\{\alpha - |V| + 1, \alpha - |V| + 2, \dots, \alpha + |V|\}$.
Therefore, $\mathcal{R}$ is a solution for $\angles{U, \mathcal{S}, k}$.

\proofsubparagraph{Soundness}

We show that if $\angles{U, \mathcal{S}, k}$ is a yes-instance then $\angles{V, \mathcal{C}}$ is satisfiable.
Let $\mathcal{R}$ be a solution for $\angles{U, \mathcal{S}, k}$.
We first show the following claim.
\begin{claim}\label{claim:isc-solution-characterization}
$\mathcal{R}$ contains exactly one of $T_i, F_i$ for each $i$.
\end{claim}
\begin{claimproof}
Let us assume that $\mathcal{R}$ does not satisfy the above condition.
Then, since $k = |V|$, there exists $i$ such that $\mathcal{R}$ does not contain both of $T_i, F_i$.
That is, $\alpha+i$ does not appear in $\mathcal{R}$ and therefore in $\bigdun \mathcal R$ there are at least $|V|-i+1$ integers greater than $\alpha+i$.
Therefore, by the definitions of $T_i$ and $F_i$, in $\bigdun \mathcal R$ there are also at least $|V|-i+1$ integers less than $\alpha-i+1$.
However, since $k = |V|$ we only have at most $2|V|$ integers in $\bigdun \mathcal R$ to cover $\{\alpha - |V| + 1, \alpha - |V| + 2, \dots, \alpha + |V|\}$.
Hence we only have at most $|V| + i - 1$ integers to cover $|V| + i$ integers $\{\alpha - i + 1, \alpha - i + 2, \dots, \alpha + |V|\}$, which is a contradiction.
\end{claimproof}
Now, we construct an assignment $\AssignmentVariable \colon V \to \{\ValueFalse, \ValueTrue\}$ that satisfies $\angles{V, \mathcal{C}}$.
We simply assign $\ValueTrue$ to $x_i$ if $T_i \in \mathcal{R}$ and $\ValueFalse$ otherwise.

By \Cref{claim:isc-solution-characterization}, in $\bigdun \mathcal R$ there are exactly $2|V|$ integers greater than $3|\mathcal{C}|$.
That is, we have $3|\mathcal{C}| \in \bigdun \mathcal R$, which implies that $\AssignmentVariable$ satisfies clause $C_{|\mathcal{C}|}$.
In addition, by the definitions of $T_i$ and $F_i$ and \Cref{claim:isc-solution-characterization}, in $\bigdun \mathcal R$ there are exactly three integers from $\{3|\mathcal{C}| - 2, 3|\mathcal{C}| - 1, 3|\mathcal{C}|\}$, which must cover them.
Therefore, we have $3|\mathcal{C}| - 3 \in \bigdun\mathcal R$, and recursively we can prove that $\AssignmentVariable$ satisfies $C_{|\mathcal{C}|-1}, C_{|\mathcal{C}|-2}, \dots, C_1$.
\qed
\end{proof}

\begin{theorem}\label{thm:hardness:diam}
    \TMC{} is \NP-complete even if the diameters of $T$ and $P$ are at least $6$ and $4$, respectively.
\end{theorem}
\begin{proof}
    It is clear that this problem is in \NP.
    We show the \NP-completeness of \TMC{} by providing a reduction from \InclusiveSetCover.
    From an instance $\angles{U, \mathcal S, k}$, we construct trees $T$ and $P$ as follows.
    We first explain how to construct $P$.
    We consider stars $R_1, \ldots, R_n$, $X_1, \ldots, X_{m-k}$, and $Y_1 \ldots, Y_k$.
    Each star $R_i$, $X_i$, and $Y_i$ have $i$, $n^3$, and $n^2$ leaves, respectively.
    Moreover, we add one vertex $p$ that connects all the centers in $R_1, \ldots, R_n$, $X_1, \ldots, X_{m-k}$, and $Y_1 \ldots, Y_k$.
    Finally, we add $3n^4$ leaves to $p$ and obtain a tree $P$ with the diameter $4$.

    We next explain how to construct $T$.
    We construct $m$ rooted trees $T_1, \ldots, T_m$ as follows.
    Let $t_i$ be the root of $T_i$ and $S_i$ be a set of integers $\{s^i_1, \ldots, s^i_\ell\}$ in $\mathcal S$.
    Each $T_i$ has $n^3$ leaves as children of $t_i$.
    For each $s^i_j$, we add the star with $s^i_j$ leaves as a child of $t_i$.
    Moreover, we add one star with $n^2$ leaves as a child of $t_i$ and
    one vertex $t$ that connects all the roots in $T_1, \ldots, T_m$.
    Finally, we add $3n^4$ leaves to $t$ and obtain a tree $T$ with the diameter $6$.
    In what follows, for each $T_i$, $R_i$, $X_i$, and $Y_i$, we denote the root of $T_i$, $R_i$, $X_i$, and $Y_i$ as $t_i$, $r_i$, $x_i$, $y_i$, respectively.
    Moreover, we denote the set of subtrees
    $\{T_1, \ldots, T_m\}$, 
    $\{R_1, \ldots, R_n\}$, 
    $\{X_1, \ldots, X_{m-k}\}$, and
    $\{Y_1, \ldots, Y_k\}$ as 
    $\mathcal T$,
    $\mathcal R$,
    $\mathcal X$, and
    $\mathcal Y$, respectively.

    \proofsubparagraph{Completeness.}
    Let $\{S_{a_1}, \ldots, S_{a_k}\}$ be
    a subset of $\mathcal S$ that has
    a surjection $f$ from $\bigdun \mathcal S$ to $U$ satisfying $v \ge f((v,i))$ for each 
    $(v, i) \in \bigdun \mathcal S$.
    In what follows, we assume that $T$ and $P$ are rooted at $t$ and $p$, respectively.
    We give an embedding $g$ from $T$ to $P$ that satisfies $g(t) = p$.
    We pick a subtree $T_{a_i}$ for each $a_i$ and define $g(t_{a_i}) = p$.
    For each integer in $S_{a_i} = \{s^{a_i}_1, \ldots, s^{a_i}_\ell\}$,
    we obtain the set of integers $\bigcup_{s \in S_{a_i}} \{f((s,i))\}$.
    From the construction of $T_{a_i}$, $T_{a_i}$ has $\ell$ stars as subtrees.
    Moreover, $j$-th star has $s^{a_i}_j$ leaves.
    Therefore, we can embed a subtree in $T_{a_i}$ with $s^{a_i}_j$ leaves into a subtree in $P$ with $f((s^{a_i}_j,i))$ leaves since $s^{a_i}_j \ge f((s^{a_i}_j,i))$.
    Moreover, for each $T_{a_i}$, we can embed one subtree in $\mathcal Y$ since $g(t_{a_i}) = p$.
    Therefore we can embed all subtrees in $P$ without each $\mathcal X$.
    For each $j \in \{1, \ldots, n\} \setminus \{a_1, \ldots, a_k\}$, $T_j$ has a subtree with $n^3$ leaves.
    Therefore, each $X$ can be embedded in each $T_j$.
    Finally, since both $t$ and $p$ have $3n^4$ neighbors with the degree $1$,
    $T$ has a $P$ as a minor.

    \proofsubparagraph{Soundness.}
    We first show that any embedding $g: T \to P$ satisfies $g(t) = p$.
    Suppose that $g(t) \neq p$.
    Since $g^{-1}(p)$ does not contain $t$, $g^{-1}(p)$ is contained in a connected component in $T - \{t\}$.
    However, each connected component has at most $n^3 + 3n^2/2$ leaves despite $p$ having $3n^4$ leaves.
    Therefore, each connected component does not contain a star with $3n^4$ leaves as a minor, and $g(t) = p$.
    In what follows, we regard $T$ and $P$ as rooted trees rooted at $t$ and $p$, respectively.

    We next show that $\mathcal T$ has $m-k$ trees $T_i$ that satisfies $g(t_i) = x$ for some $X \in \mathcal X$, where $x$ is the root of $X$.
    Since $g(t) = p$, $g^{-1}(x)$ is contained in some $T_i$.
    If $g(t_i) \neq x$, $g^{-1}(X)$ is contained in a connected component in $T_i - x$.
    However, each connected component in $T_i - x$ has at most $n^2$ leaves despite $X$ having $n^3$ leaves.
    Therefore $g(t_i) = x$ holds.
    Moreover, since $g(t) = p$ and $g(t_i) = x$, $T_i$ has no vertices $v$ such that $g(v) \not\in V(X)$.
    Since $\mathcal X$ has $m-k$ subtrees, $\mathcal T$ has $m-k$ subtrees as above.

    Let $\{T_{a_1}, \ldots, T_{a_k}\}$ be the subtrees in $\mathcal T$ that satisfies $g(t_i) \neq x$ for any $X \in \mathcal X$.
    We show that for any $T_{a_i}$, 
    either $g(t_{a_i}) = p$ or $g(t_{a_i}) = y$ for some $Y \in \mathcal Y$, where $y$ is the root of $Y$.
    If $g(t_{a_i}) \neq p$ and $g(t_{a_i}) \neq y$ for any $Y \in \mathcal Y$,
    $V(T_{a_i}) \setminus \{t_{a_i}\}$ has no vertices $v$ such that $g(v) = y$ since any $v$ does not adjacent to $t$ even if $y$ adjacent to $p$.
    Moreover,  $T_{a_i}$ contains at most one subtree in $\mathcal Y$ even if $g(t_{a_i}) = p$ or $g(t_{a_i}) = y$.
    Since $\mathcal Y$ has $k$ subtrees, any embedding satisfies either $g(t_{a_i}) = p$ or $g(t_{a_i}) = y$.

    From the above discussion, 
    for each $X \in \mathcal X$, $g^{-1}(x)$ contains a child of $t$ and
    for each $Y \in \mathcal Y$, $g^{-1}(y)$ contains a child of $t$.
    Moreover, when $g(t_i) = x$, $T_i$ has no vertex $v$ such that $g(v) \not\in V(X)$.
    Similarly, when $g(t_i) = y$, $T_i$ has no vertex $v$ such that $g(v) \not\in V(Y)$.
    Therefore, for any $R \in \mathcal R$, $g^{-1}(R)$ consists of vertices in $T_i$ satisfying $g(t_i) = p$.
    From the definition of $T_i$, $T_i - \{t_i\}$ has $|S_i| + 1$ stars.
    Since $g(t_i) = p$, $g^{-1}(R)$ is contained in a star in $T_i - \{t_i\}$.
    Therefore, the number of leaves of this star is greater than or equal to the number of leaves of $R$.
    Since $\mathcal T$ has at most $k$ subtrees such that $g(t_i) \neq x$ for any $X \in \mathcal X$,
    if we select $S_i$ if and only if $g(t_i) \neq x$ for any $X \in \mathcal X$,
    the number of sets is at most $k$.
    Moreover, since $g$ is an embedding from $T$ to $P$, these selections from $\mathcal S$ are a solution of $\angles{U, \mathcal S, k}$.
    Therefore, $\angles{U, \mathcal S, k}$ is a yes-instance if $T$ contains $P$ as a minor.
    \qed
\end{proof}

\begin{figure}
    \centering
    \includestandalone[width=\linewidth]{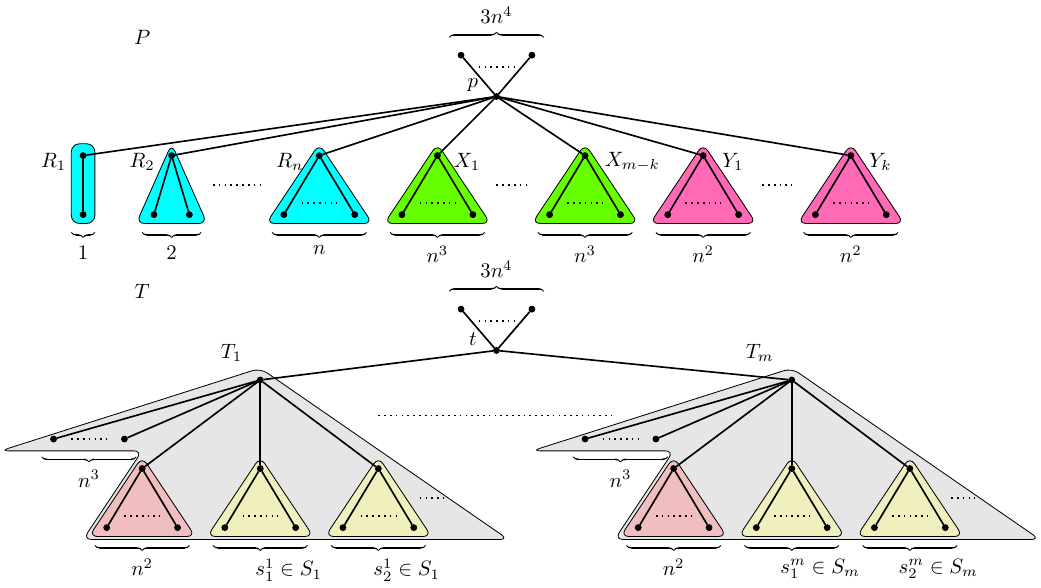}
    \caption{An example of the construction of $T$ and $P$ in the proof of~\cref{thm:hardness:diam}}
    \label{fig:TMC:hardness:construction}
\end{figure}

Since $\mathrm{pe}(T) \le k$ if $\diam(T) = 2k$, we obtain the following corollary.

\begin{corollary}
        \TMC{} is \NP-complete even if the path eccentricities of $T$ and $P$ are $3$ and $2$, respectively.
\end{corollary}

\subsection{Bounded Pathwidth}\label{subsec:hardness:pw} 
In this subsection, we show that \TMC{} is \NP-complete even if the pathwidths of $T$ and $P$ are 2.
To prove this, we first consider the following problem, which we call \PPM{}.

\defproblem{\PPM }{%
A partial ordered set $\angles{U, \le_U}$, a subset $X$ of $U^2$, and a pair $(Y, Z)$ where $Y$ is a subset of $U^2$ and $Z$ is a subset of $U$.}{
Does there exists two injections $f\colon Y \to X$ and $g \colon Z \to X \times \{1,2\}$ such that 
\begin{itemize}
    \item  $f(Y) \cap \setdef{x \in X}{(x, i) \in g(Z)} = \emptyset$,
    \item  if $f((y_1, y_2)) = (x_1, x_2)$ then $(y_1 \le_U x_1) \land (y_2 \le_U x_2)$ or $(y_2 \le_U x_1) \land (y_1 \le_U x_2)$, and 
    \item  if $g(z) = ((x_1,x_2), i)$ then $z \le_U x_i$.
\end{itemize}
}
\begin{lemmarep}
    {\PPM} is \NP-complete.
\end{lemmarep}
\newcommand{\tupvec}[1]{\boldsymbol{#1}}
\begin{proof}

\newcommand{\tval}{\mathrm{true}}
\newcommand{\fval}{\mathrm{false}}
It is clear that this problem is in NP.
We reduce from \textsc{CNF-SAT} with each clause having at most three literals, each variable appears exactly two times as a positive literal and exactly one time as a negative literal in all clauses. This problem is known to be \NP-complete~\cite{Tovey:1984:DAM}.
Let us consider an instance $\angles{V, \mathcal{C}}$ of this \textsc{SAT} variant, where $V$ is a set of variables $\{v_1, v_2, \dots, v_{n}\}$ and $\mathcal{C}$ is a set of clauses $\{C_1, C_2, \dots, C_{m}\}$.
We define $p^i_1$ and $p^i_2$ to be indices of the clause in which the variable $v_i$ appears as positive,
and $n^i$ to be the index of the clause in which the variable $v_i$ appears as negative. 
    From this instance, we reduce to an instance $\angles{\angles{U, \preceq_U}, X, \parens{Y,Z}}$ of \PPM{} in the following way.
\begin{itemize}
    \item  Let $U$ be the $\parens{\mathbb Z \cup \braces{-\infty}}^3$ and $\le$ be the natural order of $\mathbb Z$ with the least element $-\infty$.
    \item For every pairs $\parens{a_0, a_1, a_2}, \parens{b_0, b_1, b_2} \in U$, $\parens{a_0, a_1, a_2} \preceq_U \parens{b_0, b_1, b_2}$ if and only if $a_0 \le b_0$, $a_1 \le b_1$, and $a_2 \le b_2$.
    \item Let $\tupvec x_i = \parens{(i, p^i_1, -p^i_1), (-i, p^i_2, -p^i_2)}$ and $\tupvec x'_i = \parens{(i, n^i, -n^i), (-i, -\infty, -\infty)}$. The set $X$ is defined by $\bigcup_{1\le i \le n} \braces{\tupvec x_i, \tupvec x'_i}$.
    \item  Let $\tupvec y_i = \parens{(i, -\infty, -\infty), (-i, -\infty, -\infty)}$ and $Y$ be $\bigcup_{1 \le i \le n} \braces{\tupvec y_i}$.
    \item  Let $z_i = \parens{-\infty, i, -i}$ and $Z$ be $\bigcup_{1 \le i \le m} \braces{z_i}$.
\end{itemize}
The above instance can be constructed in polynomial time of the size of $\angles{V, \mathcal C}$.
From now, we show that $\angles{V, \mathcal C}$ is satisfiable if and only if $\angles*{\angles{U, \preceq}, X, \parens{Y, Z}}$ is a yes-instance.

\proofsubparagraph{Completeness.}
We show that if $\angles{V, \mathcal C}$ is satisfiable then $\angles*{\angles{U, \preceq_U}, X, \parens{Y, Z}}$ is a yes-instance.
Let $\varphi \colon V \to \braces{\fval, \tval}$ be an assignment that satisfies $\angles{V,C}$.
Then, there is a mapping $\psi\colon \mathcal C \to V$ such that 
if $\varphi(\psi(C)) = \fval$, then clause $C$ contains a variable $\psi(C)$ as a negative literal,
otherwise clause $C$ contains a variable $\psi(C)$ as a positive literal.
Moreover, we can assume that if $\varphi(v) = \fval$, then $|\psi^{-1}(v)| \le 1$, otherwise $|\psi^{-1}(v)| \le 2$.
Two mappings $f \colon Y \to X$ and $g \colon Z \to X \times \braces{1,2}$ is defined by the following:
\begin{align*}
    f(\tupvec y_i) &= \begin{cases}
        \tupvec x_i  & \parens{\varphi(v_i) = \fval} \\
        \tupvec x'_i & \parens{\varphi(v_i) = \tval}
    \end{cases},
    \\
    g(z_i) &= \begin{cases}
        (\tupvec x'_j, 1) & \parens{\psi(C_i) = v_j \land \varphi(v_j) = \fval \land i = n^j} \\
        (\tupvec x_j, 1)  & \parens{\psi(C_i) = v_j  \land \varphi(v_j) = \tval \land i = p^j_1} \\
        (\tupvec x_j, 2)  & \parens{\psi(C_i) = v_j  \land \varphi(v_j) = \tval \land i = p^j_2}
    \end{cases}.
\end{align*}
Since $p^j_1 \neq p^j_2$, two mappings $f$ and $g$ are injective.
From the constructions, $f$ and $g$ satisfy the three required conditions:
$f(Y) \cap \setdef{\tupvec x \in X}{(\tupvec x, i) \in g(Z)} = \emptyset$;
if $f((y_1, y_2)) = (x_1, x_2)$ then $y_1 \preceq_U x_1$ and $y_2 \preceq_U x_2$; and 
if $g(z_j) = ((x_1,x_2), i)$ then $z_j \preceq_U x_i$.

\proofsubparagraph{Soundness.}
We show that if $\angles*{\angles{U, \preceq_U}, X, \parens{Y, Z}}$ is a yes-instance then $\angles{V, \mathcal C}$ is satisfiable.
Then we have two injections $f$ and $g$ such that
$f(Y) \cap \setdef{\tupvec x \in X}{(\tupvec x, i) \in g(Z)} = \emptyset$;
if $f((y_1, y_2)) = (x_1, x_2)$ then  $(y_1 \le_U x_1) \land (y_2 \le_U x_2)$ or $(y_2 \le_U x_1) \land (y_2 \le_U x_1)$; and 
if $g(z_j) = ((x_1,x_2), i)$ then $z_j \preceq_U x_i$.

We first show that $f(\tupvec y_i) = \tupvec x_i$ or $f(\tupvec y_i) = \tupvec x'_i$ for all $i \in [n]$.
Let $f(\tupvec y_i) = \tupvec x_j$. From the condition of $f$ and definition of $\preceq_U$,
we have $i \le j$ and $-i \le -j$ (or $i \le -j$ and $-i \le j$ but it is not possible since $i,j>0$), and these imply $i = j$.

Furthermore, we show that the clause $C_i$ contains the variable $v_j$ as a positive literal if $g(z_i) = (\tupvec x_j, k)$ and as a negative literal if $g(z_i) = (\tupvec x'_j, k)$.
Let $g(z_i) = (\tupvec x, k)$ and $\tupvec x = ((a_1, b_1, c_1), (a_2, b_2, c_2))$.
Then $i \le b_k$ and $-i \le c_k$ since $z_i \preceq_U (a_k, b_k, c_k)$.
If $b_k = -\infty$ and $c_k = -\infty$, then $z \not \preceq_U (a_k, b_k, c_k)$, which is a contradiction.
Thus we can assume that $c_k = -b_k$.
Since $i \le b_k$ and $-i \le -b_k$, hence $i = b_k$, and this means the clause $C_i$ contains the variable $x_j$ since $b_k$ is equal to either $p^j_1$, $p^j_2$ or $n^j$ for some $j$.
If $i = p^j_1$ or $i = p^j_2$ then $\tupvec x = \tupvec x_j$, and if $i = n^j$ then $\tupvec x = \tupvec x'_j$. Therefore, the claim follows.

Finally, we define an assignment $\varphi \colon V \to \{\fval, \tval\}$ as $\varphi(v_i) = \tval$ if $f(\tupvec y_i) = \tupvec x'_i$, otherwise $\varphi(v_i) = \fval$.
Since $f(Y) \cap \setdef{\tupvec x \in X}{(\tupvec x, i) \in g(Z)} = \emptyset$,
if $g(z_i) = (\tupvec x_j, k)$ or $g(z_i) = (\tupvec x'_j, k)$ then $\tupvec x_j \notin f(Y)$ or $\tupvec x'_j \notin f(Y)$, respectively.
This, combined with the arguments above, implies that each clause $C_i$ is satisfied by the variable~$v_j$ on the assignment~$\varphi$.
\qed
\end{proof}

From here, we provide a proof of the following theorem.
\begin{theorem} \label{thm:hardness:pathwidth}
\TMC{} is \NP-complete even if the pathwidths of $T$ and $P$ are $2$.
\end{theorem}

We show the \NP-completeness by presenting a reduction from \PPM.
Let $\angles{\angles{U, \le_U}, X, \parens{Y,Z}}$ be an instance of \PPM{}.
Let $U=\braces{u_0, u_1, \dots, u_{n-1}}$.
Without loss of generality, we can assume that $X$, $Y$, and $Z$ are not empty, and $U$ contains exactly all of the elements that appear in $X$, $Y$, and $Z$.
We can also assume $2|X| = 2|Y| + |Z|$ without loss of generality, 
because creating a new element $u$ of $U$ which is smaller than any element of $U$ and adding $u$ to $Z$ does not change the solution as long as $2|X| > 2|Y| + |Z|$.
First of all, we define the following notation,
to describe an element of the partial order $\angles{U, \le_U}$ into a caterpillar.
\begin{definition}
The \emph{\ordercat} of $a \in U$ is a graph $\posetcat(a)$ such that
\begin{itemize}
    \item the vertex set is the union of $V^a = \braces{v^a_0, v^a_1, \dots, v^a_{n-1}, v^{a}_n, v^a_{n+1}}$, $L^a =\setdef{l^a_i}{u_i \le_U a}$, and
    \item the edge set is $\bigcup_{0\le i \le n} \braces{v^a_i, v^a_{i+1}} \cup \bigcup_{l^a_i \in L^a}\braces{v^a_i, l^a_i}$.
\end{itemize}
\end{definition}
An example of \ordercat{} is shown in \cref{fig:pathwidth:ordercat}.
Note that $v^a_n$ and $v^a_{n+1}$ do not correspond to any elements in $U$,
and guarantees that the maximum path length from $v_0$ in any \ordercat{} is exactly $n$.
Since the degree of every vertex $l^a_i$ is 1,
and since an edge set $\bigcup_{0\le i \le n-2} \braces{v^a_i, v^a_{i+1}}$ forms a path graph, every \ordercat{} is a caterpillar.
Note that the number of vertices is at most $2n+2$ for every order caterpillar. 
\begin{figure}[tb]
    \centering
    \includestandalone[width=\linewidth]{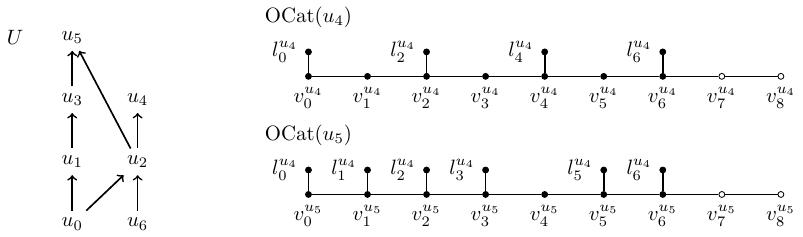}
    \caption{
        Examples of a partial order set $U=\braces{u_0, \dots, u_6}$ and the \ordercat s of $u_4$ and $u_5$.
        Partial order $\le_U$ denoted by the Hasse diagram of $\angles{U,\le_U}$, i.e.,
        an arrow from $a$ to $b$ indicates $a \le_U b$ and there is no $c$ such that $a <_U c <_U b$.
        In pictures of \ordercat s, a white node denotes a node such that there is no corresponding vertex in $U$.
    }
    \label{fig:pathwidth:ordercat}
\end{figure}
\begin{observation}
Every \ordercat{} is a caterpillar, and hence its pathwidth is 1.
\end{observation}

\begin{definition}
    Let $a,b\in U$. 
    When $l^b_i\in L^b$ if $l^a_i \in L^a$ for all $i \in \{0,\cdots, n-1\}$, 
    we can define the mapping $f\colon \posetcat(b) \to \posetcat(a)$ such that
    $f(v^b_i) = v^a_i$ for all $i \in \{0, \dots, n+1\}$ and
    $f(l^b_i) = l^a_i$ if $u_i \le_U a$, $f(l^b_i) = v^a_i$ if $u_i \not \le_U a$, and
    we call this mapping $f$ the \emph{natural embedding} from $\posetcat(b)$ to $\posetcat(a)$.
    If there exists $i$ such that $l^b_i \not \in L^b$ and $l^a_i \in L^a$, we say that the natural embedding from $\posetcat(b)$ to $\posetcat(a)$ does not exists.
\end{definition}
Clearly, for $a, b \in U$, the natural embedding from $\posetcat(b)$ to $\posetcat(a)$ is an embedding from $\posetcat(b)$ to $\posetcat(a)$ if it exists.
By the transitivity and reflexivity of $\le_U$ relation, we have the following.

\begin{observationrep}\label{lem:pathwidth:ordercat-minor}
Let $a,b\in U$. 
There exists the natural embedding from $\posetcat(b)$ to $\posetcat(a)$ if and only if $a\le_U b$. 
\end{observationrep}
\begin{proof}
Let $L(x) = \setdef{l\in U}{l \le_U x}$ for all $x\in U$. Note that an element in $L(x)$ corresponds to a vertices $L^a$ in the \ordercat{} of $a$.

Assume that $a \le_U b$. By the transitivity of $\le_U$, that is, if $x \le_U y$ and $y \le_U z$ then $x \le_U z$ for all $x,y,z \in U$, we have $L(a) \subseteq L(b)$. Hence, the natural embedding from $\posetcat(a)$ to $\posetcat(b)$ can be defined. 

Assume that there exists a natural embedding $f\colon \posetcat(b) \to \posetcat(a)$ and let $a = u_p$.
By the reflexivity of $\le_U$, that is, $u \le_U u$ for all $u\in U$, we have $u_p \le_U u_p$.
Thus $l^a_p \in L^a$, which implies $l^b_p \in L^b$ since $f^{-1}(l^a_p) = \{l^b_p\}$ by the definition of a natural embedding.
Therefore, we have $b \ge_U u_p = a$.
\qed
\end{proof}

\proofsubparagraph{Construction of \TMC{} instance $\angles{T,P}$.}
See \cref{fig:pathwidth:tp} for the whole image of \TMC{} instance $\angles{T,P}$.
For a pair $\tupvec x= (a,b)$, we write $\tupvec x_1$ for the first element $a$ and $\tupvec x_2$ for the second element $b$. 

We first define a family of trees $(T_{x})_{\tupvec x \in X}$ to describe $T$.
Let $\tupvec x\in X$. Each tree $T_{x}$ consists of three part, two subtrees $T^L_{x}$ and $T^R_{x}$, and a root vertex $r_{x}$.
A subtree $T^L_{x}$ is just $\posetcat(x_1)$, and 
$T^R_{x}$ is just $\posetcat(x_2)$.
Connect $r_x$ to $v^{x_1}_0$ in $T^L_x$ and $v^{x_2}_0$ in $T^R_x$. 
Then we obtain a family of trees $(T_x)_{x \in X}$.
Note that each $T_x$ is a caterpillar such that its backbone has $2n+5$ vertices and the number of vertices is at most $4n+5$.
Add a new vertex $r_T$ and connect $r_T$ to all $r_x$ in $T_x$, then we obtain a tree $T$.
Note that the number of vertices of $T$ is at most $(4n +5)\cdot |X|  + 1$.
Since each connected component of $T-\{r_T\}$ is a caterpillar, its pathwidth is 1, and the pathwidth of $T$ is at most 2.

We next explain how to construct $P$. 
First, we construct a family of trees $(Q_y)_{y\in Y}$ by an analogous way to $(T_x)_{\tupvec x \in X}$.
That is, for $\tupvec y \in Y$, tree $Q_y$ is a tree has a root vertex $r_y$ and two substrees $Q^L_y$ and $Q^R_y$ such that
$Q^L_y$ is just $\posetcat(y_1)$, $Q^R_y$ just $\posetcat(y_2)$, and $r_y$ is connected to $v^{y_1}_0$ in $Q^L_y$ and $v^{y_2}_0$ in $Q^R_y$.
Next, we define a family of trees $(R_z)_{z \in Z}$, where each tree $R_z$ is just $\posetcat(z)$.
Finally, we add a new vertex $r_P$ and connect $r_P$ to each $r_y$ in $Q_y$ and each $v_0^z$ in $R_z$, and then we obtain $P$.
Note that the number of vertices of $P$ is at most $(4n+5) \cdot |Y| + (2n+2) \cdot |Z|  + 1$.
Since each connected component of $P-\{r_P\}$ is a caterpillar, its pathwidth is 1, and the pathwidth of $P$ is at most 2.

\begin{figure}[tb]
    \centering
    \includestandalone[width=\linewidth]{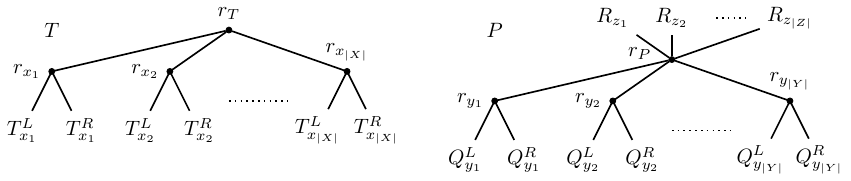}
    \caption{
        An example of the reduction in the proof of~\cref{thm:hardness:pathwidth}. 
    }
    \label{fig:pathwidth:tp}
\end{figure}

\begin{lemmarep}\label{lem:pathwidth:completeness}
If $\angles{\angles{U, \leq_U}, X, \parens{Y,Z}}$ is a yes-instance then $P$ is a minor of $T$.
\end{lemmarep}
\begin{proof}
Assume that $\angles{\angles{U, \leq_U}, X, \parens{Y,Z}}$ is a yes-instance, then there exists two injections $f\colon Y \to X$ and $g \colon Z \to X \times \{1,2\}$ such that:
$f(Y) \cap \setdef{\tupvec x \in X}{(\tupvec x, i) \in g(Z)} = \emptyset$;
if $f(\tupvec y) = \tupvec x$ then  $(y_1 \le_U x_1) \land (y_2 \le_U x_2)$ or $(y_2 \le_U x_1) \land (y_1 \le_U x_2)$; and 
if $g(z) = (\tupvec x, i)$ then $z \le x_i$. 
Since $f$ and $g$ are injective, the inverses $f^{-1}\colon f(Y) \to Y$ and $g^{-1}\colon g(Z) \to Z$ are determined uniquely and surjective.

Now, we define a mapping $\varphi\colon T \to P$ in the following way.
Let $\tupvec x \in f(Y)$ and $\tupvec y = f^{-1}(x)$.
Assume that $y_1 \le_U x_1$ and $y_2 \le_U x_2$. 
From \cref{lem:pathwidth:ordercat-minor}, 
there are two natural embeddings $\alpha\colon T_x^L \to Q_y^L$ and $\beta\colon T_x^R \to Q_y^R$.
We define the $\varphi$ for the vertices in $T_x$ as the following:
$\varphi(r_x) = r_y \in Q_y$;
$\varphi(v) = \alpha(v)$ for $v \in V(T_x^L)$;
$\varphi(v) = \beta(v)$ for $v \in V(T_x^R)$.
In the case of $y_2 \le_U x_1$ and $y_1 \le_U x_2$, we define $\varphi$ in a similar way by swapping $T_x^L$ and $T_x^R$.

Let $(\tupvec x, i) \in g(Z)$ and $z \in g^{-1}(\tupvec x,i)$.
We set $\varphi(r_x) = r_P$.
Let us consider the case $i=1$. 
Since $x_1 \le_U z$, and from \cref{lem:pathwidth:ordercat-minor}, there is a natural embedding $\gamma$ from $T^L_x$ to $R_z$. 
We set $\varphi(v) = \gamma(v)$ for all $v \in T^L_x$ and $\varphi(r_x) = r_P$.
Since $\gamma$ is an embedding, if we restrict $\varphi$ to $T_x^L$, it is clear that the restriction is an embedding from $T^L_x$ to $R_z$.
In the case of $i=2$, we define $\varphi$ in a similar way by replacing $T_x^L$ to $T_x^R$.

Finally, we set $\varphi(v) = r_P$ for all of the undefined vertex $v$ in $T$, then conclude the definition of $\varphi$.
Note that $\varphi(r_T) = r_P$. 

\begin{claim}\label{clm:pathwidth:correctemb}
    The mapping $\varphi$ is an embedding from $T$ to $P$.
\end{claim}
\begin{claimproof}
Since all of $\alpha$, $\beta$, and $\gamma$ in the definition of $\varphi$ are embeddings,
it is sufficient to consider only those vertices that are not related to them,
that is, $r_P \in V(P)$, vertices in $\varphi^{-1}(r_P) \subseteq V(T)$, $p^x$ in $T_x$, and $p^y$ in $Q_y$.

First, we show the connectivity of $\varphi^{-1}(v)$. 
It is clear that the subgraph of $T$ induced by $\varphi^{-1}(v)$ is connected for all $v \neq r_P$,
since if $v = r_y$ for some $\tupvec y\in Y$ then $\varphi^{-1}(v)$ is just a singleton, otherwise $\varphi^{-1}(v)$ is defined by $\alpha$, $\beta$, or $\gamma$.
Hence, we show that the subtree of $T$ induced by $\varphi^{-1}(r_P)$ is connected.
Let $v \in \varphi^{-1}(r_P)$.
By the definition of $\varphi$, we have $\varphi^{-1}(r_P) \subseteq V(T)\setminus \setdef{v \in V(Q_y)}{\tupvec y \in f(Y)}$.
For $v = r_x \in \varphi^{-1}(r_P)$, $v$ is connected to $r_T \in \varphi^{-1}(r_P)$.
The remaining cases, $v$ in $V(T^L_x)$ or $V(T^R_x) \cup \{p_x^R\}$, imply $V(T^L_x) \subseteq \varphi^{-1}(r_P)$ or $V(T^R_x) \cup \{p_x^R\} \subseteq \varphi^{-1}(r_P)$.
Therefore, the connectivity of $\varphi^{-1}(v)$ follows.

Finally, we show that there exists an edge $e'=\braces{u',v'}$ of $T$ for all $e = \braces{u,v} \in E(P)$ such that $\varphi(u') = u$ and $\varphi(v') = v$.
Only the following cases need to be considered: 
\begin{itemize}
    \item $e= \braces{r_y, v^{y_1}_0}$ for $\tupvec y\in Y$, let $\tupvec x = f(\tupvec y)$ then $\varphi(r_x) = r_y$, $\varphi(v^{x_1}_0) = v^{y_1}_0$ and $\{r_x, v^{x_1}_0\} \in E(T)$;
    \item $e=\braces{r_y, v^{y_2}_0}$ for $\tupvec y \in Y$, this case can be shown in a similar way to the first case;
    \item $e=\braces{r_P, r_y}$ for $\tupvec y \in Y$, let $\tupvec x = f(\tupvec y)$, then $\varphi(r_T) = r_P$, $\varphi(r_x) = r_y$ and $\braces{r_T, r_x}\in E(T)$;
    \item  $e=\braces{r_P, v^z_0}$ for $z \in Z$, let $(\tupvec x, i) = g(z)$, then $\varphi(r_x) = r_P$, $\varphi(v^{x_i}_0) = v^z_0$, and $\braces{r_x, v^{x_i}_0} \in E(T)$.
\end{itemize}
Hence, the claim follows.
\end{claimproof}
\cref{clm:pathwidth:correctemb} implies the lemma.
\qed
\end{proof}

\begin{lemmarep}\label{lem:pathwidth:soundness}
If $P$ is a minor of $T$ then $\angles{\angles{U, \leq_U}, X, \parens{Y,Z}}$ is a yes-instance.
\end{lemmarep}
\begin{proof}
Assume that $P$ is a minor of $T$, then there exists an embedding $\varphi \colon T \to P$.

We first show that $\varphi(r_T) = r_P$. Suppose that $\varphi(r_T) \neq r_P$.
Since $\varphi^{-1}(r_P)$ is connected and each connected component of $T-r_T$ is one of tree $T_x$, there exists a subtree $T_x$ such that $\varphi^{-1}(r_P) \subseteq V(T_x)$.
Fix this $T_x$ and $\tupvec x$.
Without loss of generality, we can assume that $|Y|, |X| \ge 2$.
Then, there are two edge disjoint paths with $n+3$ vertices from $r_P$,
such as $\{r_P, r_y, v^{y_1}_0, \dots, v^{y_1}_{n+1}\}$.
Let us consider the case $r_x \in \varphi^{-1}(r_P)$.
Since there is no path with $n+2$ vertices from $r_x$ which does not contain $r_T$,
there is at most one path with $n+2$ vertices from one of $\varphi^{-1}(r_P)$.
Hence $r_x \not\in \varphi^{-1}(r_P)$, which implies $\varphi^{-1}(r_P) \subseteq T^L$ or $\varphi^{-1}(r_P) \subseteq T^R$.
In neither case, there is no path with $n+2$ vertices from one of $\varphi^{-1}(r_P)$ which does not contain $r_x$.
This contradicts $\varphi(r_T) \neq r_P$.

Hence $r_T \in \varphi^{-1}(r_P)$,
which implies that each connected component of $P-r_P$ corresponds to a connected component of $T-r_T$ by $\varphi^{-1}$.
Here, each longest path of a connected component $T_x$ of $T-r_T$ is either
$\{v^{x_1}_{n+1}, \cdots, v^{x_1}_{0}, r_x, v^{x_2}_{0}, \cdots, v^{x_2}_{n+1}\}$ or the reverse.
The same is true for a connected component $Q_y$ of $P - r_P$.
Thus, every longest path of $T_x$ and $Q_y$ have $2(n+2)+1$ vertices and the central vertex are $r_x$ and $r_y$, respectively.
Hence, we have $\varphi(r_x) = r_y$.
Furthermore, $\varphi(v^{x_1}_i) = v^{y_1}_i$ and $\varphi(v^{x_2}_i) = v^{y_2}_i$, or 
$\varphi(v^{x_1}_i) = v^{y_2}_i$ and $\varphi(v^{x_2}_i) = v^{y_1}_i$.
Define $f(\tupvec y) = \tupvec x$ by the above correspondence for all $y\in Y$.
Then, from \cref{lem:pathwidth:ordercat-minor},
$\tupvec y\in Y$ and $\tupvec x = f(\tupvec y)$ are satisfied the condition $(y_1 \le_U x_1) \land (y_2 \le_U x_2)$ or $(y_2 \le_U x_1) \land (y_1 \le_U x_2)$.

Let $P'$ be the subgraph of $P$ induced by $V(P) \setminus \bigcup_{\tupvec y\in Y}V(Q_y) = \{r_P\} \cup \bigcup_{z\in Z}V(R_z)$, and
$T'$ be the subgraph of $T$ induced by $V(T) \setminus \bigcup_{\tupvec y\in Y}\varphi^{-1}(V(Q_y))$.
Since $\varphi$ is an embedding from $T$ to $P$, the restriction of $\varphi$ to $T'$ is an embedding to $P'$.
Recall that we can assume that $2|X| = 2|Y| + |Z|$ without loss of generality.
Here, there are $|Z| = 2|X| - 2|Y|$ vertex disjoint paths with $n+1$ vertices in $P'$ which starts from an adjacent vertex of $r_P$.  
Suppose that there is $r_x$ in $T'$ such that $r_x \not \in \varphi^{-1}(r_P)$.
Then, there exists $z \in Z$ such that $\varphi(V(T_x)) \subseteq V(R_z)$, from the definition of an embedding.
Now, the graph $T'-T_x$ has at most $2(|X| - |Y| -1)$ vertex disjoint paths with $n+1$ vertices.
However, the graph $P'-R_z$ has exact $2|Z| - 1$ vertex disjoint paths with $n+1$ vertices,
so $P'-R_z$ is not a minor of $T'-T_x$.
Therefore, for each $r_x$ in $T'$, $\varphi(r_x) = r_P$.
Now, there are exact $2|X| -2|Y|$ vertex disjoint paths with $n+1$ vertices that do not contain $r_T$ or any of $r_x$'s.
This implies that, for all $S = T^R_x$ or $S = T^L_x$ in $T'$, the restriction of $\varphi$ to $S$ is a natural embedding to some $R_z$.
We set $g(z) = (\tupvec x, 1)$ if the case $S = T^L_x$, and set $g(z) = (\tupvec x, 2)$ if the case $S = T^R_x$.
Then, from \cref{lem:pathwidth:ordercat-minor}, if $g(z) = (\tupvec x, i)$ then $z \le x_i$. 

From the construction of $f$ and $g$, it is clear that $f(Y) \cap \setdef{\tupvec x \in X}{(\tupvec x, i) \in g(Z)} = \emptyset$.
\qed
\end{proof}
This completes the proof of \Cref{thm:hardness:pathwidth}.

\section{Polynomial-Time Algorithms with Small Path Eccentricity and Its Application for the Other Positive Results}\label{sec:positive-results}

We give two polynomial-time algorithms for \TMC{} with a small path eccentricity.
The former algorithm determines whether a tree $T$ contains a caterpillar $P$.
The latter algorithm determines whether a lobster $T$ contains a lobster $P$.
In \Cref{subsec:cr}, we give polynomial-time algorithms for all cases in \Cref{tab:summary} using the above two algorithms.
Throughout this section, we assume $|V(P)|\geq 2$; otherwise, the problem is trivial.

\subsection{Tree-Caterpillar Containment}
We begin by considering the case where $P$ is a caterpillar. The algorithm is given in Algorithm~\ref{alg:cattree}.
It is easy to verify that Algorithm~\ref{alg:cattree} works in polynomial time.
Briefly, \Cref{alg:cattree} first guesses a backbone $C$ of $T$ ($u$-$v$ path) that corresponds to the backbone of $P$,
and then finds out how to contract $C$ to form the backbone of $P$ by a greedy method.
Moreover, the algorithm is based on the fact that contracting an internal vertices in $V(T)\setminus C$ to a vertex in $C$ does not affect whether there is a minor embedding $f\colon T \to P$ such that $f(C)$ maps to the backbone of $P$ since $P$ is a caterpillar.
Thus, it can be computed whether the backbone of $P$ can be embedded into $C$ by focusing only on the number of leaves.

\begin{algorithm}[t!]
\caption{A polynomial-time algorithm for tree-caterpillar containment.}\label{alg:cattree}
\Procedure{\emph{\Call{CatInTree}{$T,P$}}}{
    Let $B=(b_1,\dots, b_s)$ be a backbone of $P$\;
    \For{$i=1,\dots, s$}{
        Let $P_i$ be the connected component of $P - E[B]$ containing $b_i$\;
        Let $p_i$ be the number of leaves in $P_i$ other than $b_i$\;
    }
    \For{$u,v\in V(T)$}{\label{line:cattree}
        Let $C=(u=c_1,\dots,c_t=v)$ be the $u$-$v$ path in $T$\;
        \For{$i=1,\dots, t$}{
            Let $T_i$ be the connected component of $T - E[C]$ containing $c_i$\;
            Let $l_i$ be the number of leaves in $T_i$ other than $c_i$\;
        }
        $x\leftarrow 0$, $\mathrm{flag}\leftarrow \mathrm{true}$\;
        \For{$i=1, \dots, s$}{\label{line:cat_loop}
            Let $j$ be the smallest index such that $p_i\leq \sum_{k=x+1}^{j}l_k$\;
            \If{There is no such $j$}{
                $\mathrm{flag} \leftarrow \mathrm{false}$\;
                \textbf{break}\;
            }
            $x\leftarrow j$\;
        }
        \lIf{$\mathrm{flag} = \mathrm{true}$}{
            \Return yes
        }
    }
    \Return no\;
}
\end{algorithm}

\begin{theorem}\label{lem:cattree}
If $P$ is a caterpillar, Algorithm~\ref{alg:cattree} returns yes if and only if $P$ is a minor of $T$.
\end{theorem}
\begin{proof}

Assume that $P$ is embedded into $T$ by a mapping $f$. Let $e_i$ be the edge connecting $f^{-1}(b_i)$ and $f^{-1}(b_{i+1})$ for $i=1,\dots,s-1$. Then, there exists a path in $T$ in which $e_1,\dots,e_{s-1}$ appear in this order. Take a minimal such path $C=(c_1,\dots,c_t)$. Consider the case we have $u=c_1$ and $v=c_t$ in the loop starting from line~\ref{line:cattree}.

Using the integers $1=z_1,\dots,z_{s+1}=t+1$, for each $i=1,\dots,s$, we define $f^{-1}(b_i)\cap C=\{c_{z_i},\dots,c_{z_{i+1}-1}\}$. Let $T_{j',j}=\left(\bigcup_{k=j'}^{j}T_k\right)/\{c_{j'},\dots,c_j\}$ for $1\leq j'\leq j\leq s$. Then, the mapping naturally induced by $f$ embeds $P_i$ into $T_{c_{z_i},c_{z_{i+1}-1}}$. In particular, the number of leaves in $T_{c_{z_i},c_{z_{i+1}-1}}$, denoted as $\sum_{k=c_{z_i}}^{c_{z_{i+1}-1}}l_k$, is greater than or equal to $p_i$. Therefore, for $j\leq c_{z_i}$ and $c_{z_{i+1}-1}\leq j'$, it holds that $p_i\leq \sum_{k=j}^{j'}l_k$. In particular, considering the $i$-th iteration of the loop starting from line~\ref{line:cat_loop} and denoting the value of $x$ at the end of that iteration as $x_i$, it is clear by induction that $x_i\leq c_{z_{i+1}-1}$ always holds. Hence, the algorithm returns ``yes''.

Conversely, assuming that the algorithm returns ``yes'', we consider the corresponding $C=(c_1,\dots,c_t)$. We define $x_i$ as above for each $i=1,\dots,s$. Then, it holds that $p_i\leq \sum_{k=c_{z_i}}^{c_{z_{i+1}-1}}l_k$. For each $i=1,\dots,s$, let $v_{i,1},\dots,v_{i,p_i}$ be the vertices in $P_i$ other than $b_i$, and let $v'_{i,1},\dots,v'_{i,p_i}$ be $p_i$ selected leaves (not on $C$) in $T_{c_{z_i}},\dots,T_{c_{z_{i+1}}}$, respectively. Define $f(w)$ as follows:
\begin{align*}
f(w)=\begin{cases}
    v_{i,j} & \parens{w=v'_{i,j}}\\
    b_i & \parens*{w\in \bigcup_{k=c_{z_i}}^{c_{z_{i+1}-1}} V(T_k)\setminus \{v'_{i,1},\dots, v'_{i,p_i}\}}
\end{cases}.
\end{align*}
Then, $f$ is a mapping that embeds $P$ into $T$.
\qed
\end{proof}

\subsection{Lobster-Lobster Containment}
In this section, we provide a polynomial-time algorithm for \TMC{} when both $T$ and $P$ are lobsters, i.e., have path eccentricity $2$. 
The overall strategy of \Cref{alg:lobster} is the same as \Cref{alg:cattree}, first guess a backbone of $T$, and decide where to contract it to form the backbone of $P$.
However, deciding whether the (partial) minor relation holds after contracting the vertices in the guessed backbone is not 
as simple as when $P$ is a caterpillar.
This means we need to solve the following subproblem.

\newcommand{\DTMfull}{\textsc{Depth 2 Tree Minor Embedding from Lobster}}
\newcommand*{\DTM}{\textsc{D2M}}
\defproblem{\DTMfull{} (\DTM)}{%
Lobster $T$, vertex $r_T$ of $T$, and tree $P$ such that the distance of each vertex of $P$ is at most 2 from $r_P$.
}{
Is there an embedding from $P$ into $T$ such that $f(r_T) = r_P$?
}

\begin{lemmarep}
   \DTM{} can be solved in polynomial time.
\end{lemmarep}
\begin{proofsketch}
The essential case is when $r_T$ is in a backbone $C= (c_1, \dots, c_k=r_T, \dots, c_t)$ of $T$.
Since $P-r_P$ is a disjoint union of stars (here, a graph with a single vertex is also called a star), if we fix an interval of the backbone $\{c_l, \dots, c_r\} \ni r_T$ that is contracted to $r_T$, then \DTM{} can be reduced to the problem to determine whether a disjoint union of stars contains a disjoint union of stars as a minor, and this problem can be solved in polynomial time. Thus, \DTM{} can be solved by trying all intervals $\{c_l, \dots, c_r\}$.
Otherwise, except for some special cases, we can show that $r_T$ must be contracted with a vertex that is closer to the backbone of $T$, and thus \DTM{} is finally reduced to the case that $r_T$ is on a backbone of $T$.
\qed
\end{proofsketch}
    {We denote a polynomial time algorithm that computes the solution of \DTM{} by \Call{EmbedFull}{$T, r_T, P, r_P$}.}
\begin{toappendix} 
The following four lemmas analyze procedures \Call{Match}{$\cdot$}, \Call{EmbedPartial}{$\cdot$}, and \Call{EmbedFull}{$\cdot$}, respectively. Each procedure, except for \Call{Match}{$\cdot$}, uses the previous procedure as a subroutine.
\begin{algorithm}[t!]
\caption{A polynomial-time algorithm to determine whether there is an embedding $f$ that embeds $P$ into $T$ such that $f(r_T)=r_P$, where $r_T$ is assumed to be on the backbone $C$ of $T$.}\label{alg:embed1}
\Procedure{\emph{\Call{Match}{$X=(x_1,\dots, x_{|X|}), X'=(x'_1,\dots, x'_{|X'|}), a\geq 0, a'\geq 0$}}}{
    Let $K=\sum_{i=1}^{|X'|}x'_i$, $h\leftarrow 0$\;
    \For{$i=1,\dots, |X|$}{\label{line:match_loop}
        \lIf{No index $h'>h$ satisfies $x_i\leq x'_{h'}$}{
            \Return no
        }
        \Else{
            Let $h'$ be the smallest index with $h'>h$ and $x_i\leq x'_{h'}$\;
            $K\leftarrow K - x'_{h'}$, $h\leftarrow h'$\;
        }
    }
    \lIf{$a\leq K+a'$}{
        \Return yes
    }
    \Return no\;
}
\Procedure{\emph{\Call{EmbedPartial}{$T,r_T\in V(T),C=(c_1,\dots, c_k=r_T, \dots, c_t),P,r_P\in V(P)$}}}{
    Let $k$ be the index with $c_k=r_T$\;
    Let $X$ be an empty multiset and $a\leftarrow 0$\;
    \For{$v\in N_P(r_P)$}{
        \lIf{$\deg_P(v)=1$}{
            $a\leftarrow a+1$
        }
        \lElse{
            $X\leftarrow X\cup \{\deg_P(v)-1\}$
        }
    }
    Sort $X$ in ascending order and denote as $(x_1, \dots, x_{|X|})$\;
    \For{$y=1,\dots, k$}{\label{line:embed1_yz1}
        \For{$z=k,\dots, t$}{\label{line:embed1_yz2}
            Let $X_{y,z}$ be an empty multiset and $a_{y,z}\leftarrow 0$\;
            \For{$v\in N_T(\{c_y,\dots, c_z\})\setminus \{c_y,\dots, c_z\}$}{
                \lIf{$\deg_T(v)=1$}{
                    $a_{y,z}\leftarrow a_{y,z}+1$
                }
                \Else{
                    Let $l_v$ be the number of leaves (other than $v$) in the connected component of $T - \{c_y,\dots, c_z\}$ containing $v$, and then $X_{y,z}\leftarrow X_{y,z}\cup \{l_v\}$\;
                }
            }
            Sort $X_{y,z}$ in ascending order and denote as $(x_{y,z,1},\dots, x_{y,z,|X_{y,z}|})$\;
            \lIf{\Call{Match}{$X,X_{y,z},a,a_{y,z}$}}{
                \Return yes
            }
        }
    }
    \Return no;
}
\end{algorithm}

\begin{lemmarep}\label{lem:lobster_match}
Let $X = (x_1, \dots, x_{|X|})$ and $X' = (x'_1, \dots, x'_{|X'|})$ be sorted sequences of integers in ascending order, and let $a, a' \in \mathbb{Z}_{\geq 0}$. In this case, \Call{Match}{$X, X', a, a'$} returns ``yes'' if and only if there exists an injective function $g$ from $\{1, \dots, |X|\}$ to $\{1, \dots, |X'|\}$ such that for all $i \in \{1, \dots, |X|\}$, $x_i \leq x'_{g(i)}$ and $a \leq a' + \sum_{i \in \{1, \dots, |X'|\} \setminus \{g(1), \dots, g(|X|)\}} x'_i$.
\end{lemmarep}
\begin{proof}

Let $g^*$ be the injective function from $\{1, \dots, |X|\}$ to $\{1, \dots, |X'|\}$ that satisfies $x_i \leq x'_{g(i)}$ for all $i \in \{1, \dots, |X|\}$ and maximizes the value of $c(g) \colon= \sum_{i \in \{1, \dots, |X'|\} \setminus \{g(1), \dots, g(|X|)\}} x'_i$. We prove that $K = c(g^*)$ holds at the end of the loop starting at line~\ref{line:match_loop}.

Assume that $g^*(i) > g^*(j)$ for some $i < j$. If we exchange the values of $g^*(i)$ and $g^*(j)$, the conditions $x_i \leq x'_{g^*(i)}$ and $x_j \leq x'_{g(j)}$ are still satisfied, and the value of $c(g^*)$ remains unchanged. Therefore, we can assume that the sequence $(g^*(1), \dots, g^*(|X|))$ is strictly increasing. Let $h_i$ be the value of $h'$ taken in the $i$-th iteration of the loop starting at line~\ref{line:match_loop}. Since $(g^*(1), \dots, g^*(|X|))$ is strictly increasing, it can be shown by induction that $h_i \leq g^*(i)$. Thus, by taking $K$ at the end of the loop, we have
\begin{align*}
K=\sum_{i\in \{1,\dots, |X'|\}\setminus \{h_1,\dots, h_{|X|}\}}x'_i
\geq \sum_{i\in \{1,\dots, |X'|\}\setminus \{g^*(1),\dots, g^*(|X|)\}}x'_i=c(g^*)
\end{align*}
and since $g^*$ is defined as maximizing $c(g)$, we have $K = c(g^*)$.
\qed
\end{proof}

\begin{lemmarep}\label{lem:lobster_partial}
Let $T$ be a lobster, $r_T$ be a vertex of $T$, $C = (c_1, \dots, c_t)$ be a backbone of $T$ that contains $r_T$. Let $P$ and $r_P$ be a tree and its vertex, respectively, such that all vertices of $P$ are at a distance of $2$ or less from $r_P$. In this case, \textsc{EmbedPartial}($T, r_T, C, P, r_P$) returns ``yes'' if and only if there exists an embedding $f$ of $P$ into $T$ such that $f(r_T) = r_P$.
\end{lemmarep}
\begin{proof}

\begin{figure}[tb]
    \centering
    \includestandalone[width=\linewidth]{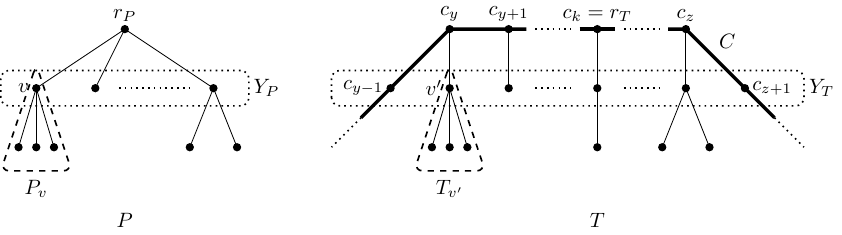}
    \caption{Trees in Lemma~\ref{lem:lobster_partial}. The left and right figure represents the tree $P$ and $T$, respectively. The thick line in $T$ represents its backbone $C$.}
    \label{fig:lobster_partial}
\end{figure}

Assuming the existence of an embedding $f$ that satisfies the conditions of the lemma, and let $C \cap f^{-1}(r_P) = {c_y, \dots, c_z}$. Since $f(r_T) = f(c_k) = r_P$, we have $y \leq k \leq z$ and thus there exists a pair of $y$ and $z$ that satisfy the double loops starting at line~\ref{line:embed1_yz1} and line~\ref{line:embed1_yz2}. Let $Y_P = N_P(r_P)$ and $Y_T = N_T({c_y, \dots, c_z})$. For each $v \in Y_P$, let $P_v$ denote the connected component of $P - {r_P}$ containing $v$ and consider it a rooted tree with $v$ as the root. Similarly, for each $v \in Y_T$, let $T_v$ denote the connected component of $T - \braces{c_y, \dots, c_z}$ that contains $v$ and consider it as a rooted tree with $v$ as the root.

Since $f$ is an embedding, for each $v \in Y_P$, there exists a unique $v' \in Y_T$ such that $f^{-1}(P_v) \subseteq T_{v'}$. Let us denote this mapping from $v \in Y_P$ to $v' \in Y_T$ as $g$. We will prove the following claim:

\begin{claim}\label{cl:loblob_g}
For a given $v' \in Y_T$, one of the following conditions holds:
\begin{itemize}
\item There is at most one $v \in Y_P$ such that $g(v) = v'$.
\item For all $v \in Y_P$ such that $g(v) = v'$, $P_v$ consists of a single vertex.
\end{itemize}
\end{claim}
\begin{claimproof}
Let us consider cases based on the depth of $T_{v'}$. If the depth of $T_{v'}$ is at least $2$, noting that $C$ is the backbone of $T$, we observe that $v' = c_{y-1}$ or $v' = c_{z+1}$. According to the choices of $y$ and $z$, it follows that $f(c_{y-1}) \neq r_P$ and $f(c_{z+1}) \neq r_P$, so $|g^{-1}(v')| = 1$.

Now, let us consider the case when the depth of $T_{v'}$ is $1$. If $f(v') = r_P$, for all $v \in Y_P$ such that $g(v) = v'$, $P_v$ consists of a single vertex. If $f(v') \neq r_P$, then $|g^{-1}(v')| \leq 1$. Additionally, when $T_{v'}$ consists of a single vertex, $|g^{-1}(v')| \leq 1$ as well.
\end{claimproof}

Let us prove that \textsc{EmbedPartial}($T,r_T,C,P,r_P$) returns ``yes.''
From the choice of $r_P$, it follows that the depth of $P_v$ for any $v\in Y_P$ is at most $1$. 
Let $v\in Y_P$. By Claim~\ref{cl:loblob_g}, if the depth of $P_v$ is $1$, we have $V(T_{g(v)})\cap Y_P=\{v\}$, and thus $\deg_P(v)-1\leq l_{g(v)}$ holds. Let $Z$ be the set of $v'\in Y_T$ that can be written as $v'=g(v)$ using such $v$.
If $P_v$ consists of a single vertex, we have $g(v)\in Y_T\setminus Z$. Hence, $a\leq a_{y,z}+\sum_{v'\in Y_T\setminus Z}l_{v'}$ holds. Therefore, there exists an injective mapping $g'$ from $X$ to $X_{y,z}$ such that $x\leq g'(x)$ for all $x\in X$ and $a\leq a_{y,z}+\sum_{x'\in X_{y,z}\setminus g'(X)}x'$, as guaranteed by Lemma~\ref{lem:lobster_match}. Consequently, \textsc{Match}($X,X_{y,z},a,a_{y,z}$) returns ``yes,'' and therefore, \textsc{EmbedPartial}($T,r_T,C,P,r_P$) also returns ``yes.''

Conversely, assuming that the algorithm returns ``yes,'' by Lemma~\ref{lem:lobster_match}, there exists $y\leq k\leq z$ and an injective mapping $g'$ from $X$ to $X_{y,z}$ such that $x\leq g'(x)$ holds for all $x\in X$ and $a\leq a_{y,z}+\sum_{x'\in X_{y,z}\setminus g'(X)}x'$. Let $Y_P$, $Y_T$, $P_v$, and $T_v$ be defined as before, and construct a mapping $f\colon T\ni u\mapsto u'\in P$ as follows.

First, let us define $f^{-1}(u)$ for $u\neq r_P$. Let $v\in Y_P$ and assume the depth of $P_v$ is $1$. Then, a corresponding $v'\in Y_T$ is determined naturally by the mapping $g'$. Let $Z$ be the set of such $v'$. Consider $v$'s children as $v_1,\dots, v_{\deg_P(v)-1}$ and the leaves of $T_{v'}$ as $v'_1,\dots, v'_{l(v')}$. Then, $\deg_P(v)-1\leq l(v')$ holds due to the conditions satisfied by $g'$. For $w\in \bigcup_{v'\in Z}V(T_{v'})$, let
\begin{align*}
f(w)=\begin{cases}
    v_i & (w=v'_i, i=1, \dots, \deg_P(v)-1)\\
    v & (\text{otherwise})
\end{cases}.
\end{align*}

There are $a$ vertices $v_1, \dots, v_a$ such that $P_v$ consists of a single vertex. According to the conditions satisfied by $g'$, at least $a$ leaves of $T$ are included in $\bigcup_{v'\in Y_T\setminus Z}V(T_{v'})$. Let $v'_1, \dots, v'_a$ be $a$ such leaves. For $w\in \bigcup_{v'\in Y_T\setminus Z}V(T_{v'})$, let
\begin{align*}
f(w)=\begin{cases}
    v_i & (w=v'_i)\\
    r_P & (\text{otherwise})
\end{cases}.
\end{align*}

Finally, let us define $f(w)=r_P$ for $w\in {c_y, \dots, c_z}$. Then, we obtain an embedding $f$ of $P$ into $T$.
\qed
\end{proof}

\begin{algorithm}[t!]
\caption{A polynomial-time algorithm to determine whether there is an embedding $f$ that embeds $P$ into $T$ such that $f(r_T)=r_P$, where $r_T$ is not necessarily on a backbone of $T$.}\label{alg:sub}
\Procedure{\emph{\Call{EmbedFull}{$T,r_T,P,r_P$}}}{
    \If{There exists a backbone $C=(c_1,\dots, c_t)$ of $T$ that contains $r_T$}{
        \Return \Call{EmbedPartial}{$T,r_T,C,P,r_P$}
    }
    \ElseIf{There exists a backbone of $T$ that contains a neighbor $c$ of $r_T$}{
        \If{All neighbor of $r_P$ in $P$ is a leaf}{
            \lIf{$\deg_P(r_P)\leq \deg_T(r_T)$}{
                \Return yes\label{line:sub_yes1}
            }
        }
        \If{There is exactly one non-leaf neighbor $q$ of $r_P$ in $P$}{
            Let $l$ be the number of leaves in the connected component of $T - r_T$ that contains $c$\;
            \lIf{$\deg_P(r_P)-1\leq \deg_T(r_T)-1$ and $\deg_P(q)-1\leq l$}{
                \Return yes\label{line:sub_yes2}
            }
        }
        \Return \Call{EmbedFull}{$T/\{r_T,c\}, r_T ,P, r_P$}\;
    }
    \Else{
        \If{$\deg_P(r_P)=1$}{
            Let $q_P$ be the unique neighbor of $r_P$ in $P$\;
            Let $l$ be the number of leaves in $T$ other than $r_T$\;
            \lIf{$\deg_P(q_P)-1\leq l$}{
                \Return yes\label{line:sub_yes3}
            }
        }
        Let $q_T$ be a unique neighbor of $r_T$ in $T$\;
        \Return \Call{EmbedFull}{$T/\{r_T,q_T\}, r_T, P,r_P$}\;
    }
}
\end{algorithm}

\begin{lemmarep}\label{lem:lobster_full}
If $T$ is a lobster, $r_T$ is a vertex of $T$, $P$ is a tree with a vertex $r_P$, and all vertices of $P$ are at a distance of $2$ or less from $r_P$, then \textsc{EmbedFull}($T,r_T,P,r_P$) returns ``yes'' if and only if there exists an embedding $f$ of $P$ into $T$ such that $f(r_T)=r_P$.
\end{lemmarep}
\begin{proof}
For each $v\in N_P(r_P)$, let $P_v$ be the connected component of $P - \{r_P\}$ that contains $v$, and consider $P_v$ as a rooted tree with $v$ as the root.
Similarly, for each $v\in N_T(r_T)$, let $T_v$ be the connected component of $T - \{r_T\}$ that contains $v$, and consider $T_v$ as a rooted tree with $v$ as the root.
If there exists a backbone of $T$ that includes $r_T$, Lemma~\ref{lem:lobster_partial} proves the lemma.

Next, let us consider the case where there is no backbone of $T$ that includes $r_T$, but a backbone of $T$ includes a neighbor $c$ of $r_T$. We will consider two cases based on whether $f(c)=r_P$ or not. First, we consider the case where $f(c)=r_P$.
Satisfying the conditions of the lemma and having $f(c) = r_P$ is equivalent to the condition that the embedding induced by $f$ embeds $P$ into $T/\{r_T,c\}$. The new vertex formed by contracting $r_T$ and $c$ lies on the backbone of $T/\{r_T,c\}$, so this case reduces to the previous case.

Next, let us consider the case where $f(c)\neq r_P$. Assuming the existence of an $f$ satisfying the conditions of the lemma and $f(c)\neq r_P$, it follows that there is at most one $v\in N_P(r_P)$ such that $f(P_v)\cap T_c\neq \emptyset$. Furthermore, due to the condition that $T$ is a lobster, $T_{v'}$ consists of a single vertex for $v'\in N_T(r_T)\setminus \{c\}$. Therefore, there is at most one $v\in N_P(r_P)$ with $P_v$ consisting of more than one vertex.

If $P_v$ consists of one vertex for all $v\in N_P(r_P)$, then noting that $f(c)\neq r_P$, no two vertices in $N_P(r_P)$ are embedded into the same $T_{v'}$ for $v'\in N_T(r_T)$. Thus, $\deg_P(r_P)\leq \deg_T(r_T)$ holds, and the algorithm returns ``yes'' at line~\ref{line:sub_yes1}. Conversely, if the algorithm returns ``yes'' at line~\ref{line:sub_yes1}, it implies $\deg_P(r_P)\leq \deg_T(r_T)$. In this case, we can take the vertices of $N_P(r_P)$ and $N_T(r_T)$ as $v_1,\dots, v_{\deg_P(r_P)}$ and $c=v'_1,\dots, v'_{\deg_T(r_T)}$, respectively, and let
\begin{align*}
f(w)=\begin{cases}
    v_1 & (w\in V(T_c))\\
    v_i & (w=v'_i, i=2,\dots, \deg_P(r_P))\\
    r_P & (\text{otherwise})
\end{cases}.
\end{align*}
Then, $f$ satisfies the required conditions.

If there exists exactly one $v\in N_P(r_P)$ such that $P_v$ consists of more than one vertex, let $q$ be such a vertex. From the previous analysis, we know that $f(c)=q$, where $c$ is the neighbor of $r_T$, and the number of leaves $l$ in $T_c$ (excluding $c$) is at least $\deg_P(q)-1$. Moreover, since $f$ embeds all vertices in $N_P(r_P)\setminus \{q\}$ to distinct vertices in $N_T(r_T)\setminus \{c\}$, we have $\deg_P(r_P)-1\leq \deg_T(r_T)-1$. Therefore, the algorithm returns ``yes'' at line~\ref{line:sub_yes2}.

Conversely, if the algorithm returns ``yes'' at line~\ref{line:sub_yes2}, it implies that $\deg_P(q)-1\leq l$ and $\deg_P(r_P)-1\leq \deg_T(r_T)-1$. Let $v_1,\dots, v_{\deg_P(q)-1}$ be the vertices of $N_P(q)\setminus \{r_P\}$, and choose $\deg_P(q)-1$ leaves from $T_c$ as $v'_1,\dots, v'_{\deg_P(q)-1}$. Additionally, let $u_1,\dots, u_{\deg_P(r_P)-1}$ be the vertices of $N_P(r_P)\setminus \{q\}$ and $u'_1,\dots, u'_{\deg_T(r_T)-1}$ be the vertices of $N_T(r_T)\setminus \{c\}$. Let
\begin{align*}
f(w)=\begin{cases}
    v_i & (w=v'_i)\\
    q & (w\in V(T_c)\setminus \{v'_1,\dots, v'_{\deg_P(v)-1}\})\\
    u_i & (w=u'_i, i=2,\dots, \deg_P(r_P)-1)\\
    r_P & (\text{otherwise})
\end{cases}.
\end{align*}
Then, $f$ satisfies the required conditions.

Finally, let us consider the case where there is no backbone of $T$ that includes neither $r_T$ nor any neighbor of $r_T$. Due to the condition that $T$ is a lobster, we have $\deg_T(r_T)=1$ in this case. Let $q_T$ be the unique neighbor of $r_T$, and consider two cases based on whether $f(q_T)=r_P$ or not. First, consider the case $f(q_T)=r_P$. Satisfying the conditions of the lemma and having $f(q_T)=r_P$ is equivalent to the condition that the embedding induced by $f$ embeds $P$ into $T/\{r_T,q_T\}$. The new vertex formed by contracting $r_T$ and $q_T$ has neighbors on the backbone of $T/\{r_T,q_T\}$, so this case reduces to the previous case.

Finally, let us consider the case where $f(q_T)\neq r_P$. Assuming the existence of an $f$ satisfying the conditions of the lemma and $f(q_T)\neq r_P$, we obtain that there is at most one $v\in N_P(r_P)$ such that $f(P_v)\cap T_{q_T}\neq \emptyset$, and combining this with $N_T(r_T)=\{q\}$, we have $\deg_P(r_P)=1$. Let $q_P$ be the unique neighbor of $r_P$. In this case, we have $f(q_T)=q_P$, and the number of leaves in $T$ (other than $r_T$) is equal to the number of leaves in $T_{q_T}$, which is denoted as $l$, and is at least $\deg_P(q_P)-1$. Thus, the algorithm returns ``yes'' at line~\ref{line:sub_yes3}.

Conversely, if the algorithm returns ``yes'' at line~\ref{line:sub_yes3}, it implies that $\deg_P(q_P)-1\leq l$. Let $v_1,\dots, v_{\deg_P(q_P)-1}$ be the vertices in $N_P(q_P)\setminus \{r_P\}$ and take $\deg_P(q_P)-1$ leaves from $T_{q_T}$ as $v'_1,\dots, v'_{\deg_P(q_P)-1}$. Let
\begin{align*}
f(w)=\begin{cases}
    v_i & (w=v'_i)\\
    r_P & (w=r_T)\\
    q_P & (\text{otherwise})
\end{cases}.
\end{align*}
Then, $f$ satisfies the required conditions.
\qed
\end{proof}
\end{toappendix} 

\begin{algorithm}[t!]
\caption{An algorithm for the case that both of trees are lobsters}\label{alg:lobster}
\Procedure{\emph{\Call{LobInLob}{$T,P$}}}{
    Let $B=(b_1,\dots, b_s)$ be a backbone of $P$\;
    \For{$i=1,\dots, s$}{
        Let $P_i$ be the connected component of $P - E[B]$ containing $b_i$\;
    }
    \For{$u,v\in V(T)$}{\label{line:lobster_uv}
        Let $C=(u=c_1,\dots,c_t=v)$ be the $u-v$ path in $T$\;
        \For{$i=1,\dots, t$}{
            Let $T_i$ be the connected component of $T- E[C]$ containing $c_i$\;
        }
        $x\leftarrow 0$, $\mathrm{flag}\leftarrow \mathrm{true}$\;
        \For{$i=1, \dots, s$}{\label{line:lob_loop}
            Let $j$ be the smallest index such that \Call{EmbedFull}{$T_{x+1,j}, c_{x+1}, P_i, b_i$} returns yes, where $T_{x+1,j}=\left(\bigcup_{k=x+1}^{j}T_k\right)/\{c_{x+1},\dots, c_j\}$\;
            \If{There is no such $j$}{
                $\mathrm{flag} \leftarrow \mathrm{false}$\;
                \textbf{break}\;
            }
            $x\leftarrow j$\;
        }
        \lIf{$\mathrm{flag} = \mathrm{true}$}{
            \Return yes
        }
    }
    \Return no\;
}
\end{algorithm}

We present an algorithm that computes the solution of \TMC{} when both trees are lobsters by using \Call{EmbedFull}{$\cdot$} as a subroutine.

\begin{theoremrep}\label{lem:loblob}
\Cref{alg:lobster} returns yes if and only if $P$ is a minor of $T$.
\end{theoremrep}
\begin{proof}
Suppose that $P$ is embedded into $T$ by the embedding $f$. For each $i=1,\dots, s-1$, let $e_i$ be the edge connecting $f^{-1}(b_i)$ and $f^{-1}(b_{i+1})$. Then, there exists a path in $T$ such that $e_1,\dots, e_{s-1}$ appear in that order along the path. We choose the minimal path among such paths and denote it as $C=(c_1,\dots, c_t)$. In the loop starting at line~\ref{line:lobster_uv}, we consider the case where $u=c_1$ and $v=c_t$.

Using integers $1=z_1,\dots, z_{s+1}=t+1$, we define $f^{-1}(b_i)\cap C=\{c_{z_i},\dots, c_{z_{i+1}-1}\}$ for $i=1,\dots, s$. By setting $T_{j',j}=\left(\bigcup_{k=j'}^{j}T_k\right)/\{c_{j'},\dots, c_j\}$ for $1\leq j'\leq j\leq s$, \textsc{EmbedFull}($T_{c_{z_i},c_{z_{i+1}-1}},c_{{z_i}},P_i,b_i$) returns ``yes''. Since adding vertices to $T_{c_{z_i},c_{z_{i+1}-1}}$ does not exclude $P_i$ as a minor, \textsc{EmbedFull}($T_{c_{j},c_{j'}},c_{j},P_i,b_i$) returns ``yes'' for $j'\leq c_{z_i}$ and $c_{z_{i+1}-1}\leq j$. In particular, by defining $x_i$ as the value of $x$ at the end of the $i$-th iteration of the loop starting at line~\ref{line:lob_loop}, it can be proven by induction that $x_i\leq c_{z_{i+1}-1}$ holds. Therefore, the algorithm returns ``yes''.

Conversely, we assume that the algorithm returns ``yes'' and consider the corresponding $C=(c_1,\dots, c_t)$. We define $x_i$ as before. By \textsc{EmbedFull}($T_{x_i,x_{i+1}-1}, c_{x_i}, P_i, b_i$), we are guaranteed the existence of an embedding $f_i$ of $P_i$ into $T_{x_i,x_{i+1}-1}$. Let
\begin{align*}
f(w)=\begin{cases}
    f_i(w) & (w\in V\left(T_{x_i,x_{i+1}-1}\right)\setminus \{c_{x_i},\dots, c_{x_{i+1}-1}\})\\
    b_i & (w\in \{c_{x_i},\dots, c_{x_{i+1}-1}\})
\end{cases}.
\end{align*}
Then, $f$ is an embedding of $P$ into $T$.
\qed
\end{proof}

\subsection{Applications of \cref{alg:cattree,alg:lobster}}\label{subsec:cr}
As shown in \Cref{tab:summary}, we give polynomial-time algorithms for \TMC{} with small diameter, path eccentricity, and pathwidth.
Since we already show the case $\mathrm{pe}(P)\le 1$ in \Cref{lem:cattree} and $\mathrm{pe}(T) \le 2$ in \Cref{lem:loblob},
we show the cases with small diameters and pathwidths.
These results can be easily shown by using the results in previous subsections.

\begin{theorem}
    \TMC{} can be solved in polynomial time when $\diam(P)\leq 3$ or $\diam(T)\leq 5$.
\end{theorem}
\begin{proof}
    Since a tree with a diameter at most $3$ is a caterpillar, 
    we can solve \TMC{} when $\diam(P) \le 3$ from \Cref{lem:cattree}.
    Moreover, when $\diam(P) > \diam(T)$, $T$ does not contains $P$ as a minor obviously.
    Therefore, we can assume that $\diam(P) \le \diam(T)$.
    Since a tree with a diameter at most $5$ is a lobster, 
    \TMC{} can be solved in polynomial time when $\diam(P) \le \diam(T) \le 5$ from \Cref{lem:loblob}.
    \qed
\end{proof}

\begin{theorem}
    \TMC{} can be solved in polynomial time when $\mathrm{pw}(P) \le 1$.
\end{theorem}
\begin{proof}
    Since a tree with pathwidth 1 is a caterpillar, 
    we obtain a polynomial-time algorithm from \Cref{lem:cattree}.
    \qed
\end{proof}
\bibliographystyle{splncs04}
\bibliography{main.bib}
\end{document}